\documentclass[11pt]{article}
\usepackage[latin1]{inputenc}
\usepackage[english]{babel}
\usepackage[namelimits]{amsmath}
\usepackage{amssymb, amsxtra, mathrsfs, bbm}
\usepackage{amsmath}
\usepackage{amsthm}

\theoremstyle{definition}
\newtheorem{Definition}{Definition}
\theoremstyle{plain}

\newtheorem{Proposition}[Definition]{Proposition}
\newtheorem{Lemma}[Definition]{Lemma}

\def\cA{{\cal A}}

\def\cC{{\cal C}}
\def\cD{{\cal D}}
\def\cE{{\cal E}}
\def\cF{{\cal F}}

\def\cH{{\cal H}}

\def\cK{{\cal K}}
\def\cL{{\cal L}}

\def\cO{{\cal O}}
\def\cP{{\cal P}}

\def\cS{{\cal S}}

\def\bR{{\mathbb R}}

\def\sD{\mathscr{D}}

\def\sH{\mathscr H}

\newcommand{\abs}[1]{\lvert#1\rvert}

\begin{document}
\title{Building non commutative spacetimes at the Planck length for Friedmann flat cosmologies\footnotetext{\hspace{-.65truecm}Work supported by the {\bf ERC Advanced Grant} 227458 OACFT
  \emph{Operator Algebras \!\&\! Conformal Field Theory}.}}
\author{{Luca Tomassini\thanks{Dipartimento di Matematica,
Universit\`a di Roma ``Tor Vergata'',
Via della Ricerca Scientifica, 1, I-00133 Rome, Italy
E-mail: {\tt tomassin@axp.mat.uniroma2.it}} 
\, and 
{Stefano Viaggiu\thanks{Dipartimento di Matematica,
Universit\`a di Roma ``Tor Vergata'',
Via della Ricerca Scientifica, 1, I-00133 Roma, Italy.
E-mail: {\tt viaggiu@axp.mat.uniroma2.it}.}}}}
\date{\today}\maketitle
\begin{abstract}
We propose physically motivated
spacetime uncertainty relations (STUR) for  flat
Friedmann-Lema\^{i}tre cosmologies.
We show that the physical features of these STUR crucially depend on whether a particle horizon is present or not. In particular, when this is the case we deduce the existence of a maximal value for the Hubble rate (or equivalently for the matter density), thus providing an indication that quantum effects may rule out a pointlike big bang singularity. 
Finally, we costruct a concrete realisation of the corresponding quantum Friedmann spacetime in terms of operators on a Hilbert space. 
\end{abstract}


{\it In loving memory of Francesco Saverio de Blasi, mathematician and friend.}

\section{Introduction}

According to General Relativity (GR), spacetime at sufficiently large scales can be described as a pseudo-Riemanian (classical) manifold $M$ locally modeled on Minkowski space with a metric $g_{\mu,\nu}=\text{diag}\{-1,1,1,1\}$, $\mu,\nu=0,\dots,3$. On the other hand, quantum matter in spacetime is described in terms of quantum fields on this manifold. However, it is a widespread idea that at sufficiently small scales the geometry of space-time becomes ``quantum'' or non commutative and loses its classical meaning. The relevant scale was recognised to be that of Planck's length
\begin{equation*}
\lambda_P=\sqrt{\frac{G\hbar}{c^3}}\simeq 1,6.10^{-33}cm,
\end{equation*}
$c$ and $G$ being respectively the speed of light and the gravitational constant. It is an interesting fact that historically one of the main motivations of this belief was the hope that the ultraviolet divergences that plague Quantum Field Theory (QFT) could be get rid of by making points ``fuzzy'' (Heisenberg himself suggested this possibility). However, not very much later the idea that this phenomenon could be rooted in intrinsic operational limitations to localisation measurements coming from the very first principles of General Relativity (GR) and non relativistic 
Quantum Mechanics (QM) emerged (see \cite{Mead}).

Later on, the emergence of Non Commutative Geometry \cite{connes} thanks to the work of Connes, Takesaki and many others has led to an interpretation of these facts in the following terms: spacetime should be considered as a non commutative manifold, and matter should be described by quantum fields on this non commutative manifold. A great amount of work has been dedicated to the task of defining and studying such new objects.
In this context, the idea that the structure of this non commutative spacetime should be derived by (or at least not be in contrast with) established physical principles was advocated by Doplicher, Fredenhagen and Roberts (DFR) in \cite{dop}\footnote{see also the interesting reviews \cite{dop3} -\cite{Pia}}. Roughly speaking, the  idea is as follows. To make observations in a given spacetime region with sides $\Delta x^\mu$, $\mu=0,\dots,3$, we must use (say) radiation with comparable wavelength. In case the localisation region is very small (so that the energy density can be assumed very big) GR tells us that this eventually lead to the formation of an event horizon, i.e. of a  black hole 
(\cite{Hawking1975}). Thus, no information could come out of the region itself, making the measurement meaningless.

In this paper we deal with the problem of extending this method to the case of flat Robertson-Friedmann-Walker cosmological spacetimes. To this end, we propose a  
modification of the isoperimetric conjecture, initially formulated  
for asymptotically flat spacetimes and proved for spherical collapses in \cite{Bi}, for Friedmann open flat cosmologies 
motivated by results of 
\cite{Br,Mal}, and use it together with Heisenberg's Principle to obtain STUR in this situation.\\
In section 2 we fix some notation and present our method for minkowskian backgrounds. In section 3 we discuss the condition for the formation of trapped surfaces in Friedmann flat cosmologies, while in section 4 our STUR are obtained. In the sections 5  we derive some physically relevant consequences of our STUR. In section 6 we consider inflationary cosmologies. Section 7 is devoted to the construction of a concrete model of quantum Friedmann spacetime. The mathematical tools necessary achieve this result are relegated to the Appendix.

\section{Preliminaries}

In this section we introduce a more thorough description of our procedure and some notation to be used later. In particular, we stress the quantum field theoretic nature of our approach, as opposed to methods based on non relativistic Quantum Mechanics\footnote{See for example \cite{Mead} for one of the first attempts, \cite{sabine} for a review and some bibliography, \cite{res} for a first assessment of the differences between the two points of view.}. For the sake of clarity and completeness, we also briefly recall previous results for Minkowski backgrounds.

\subsection{Localization and the notion of space-time point}\label{sub:local}

In classical general relativity one deals with a four dimensional manifold (spacetime) $M$ equipped with a lorentzian metric. All points in a given spacetime can be labeled by (local) coordinates $x^\mu$. Given this structure, \textit{relativistic} quantum matter is introduced by means of quantum \textit{fields}. Historically, since gravitational effects are negligible at the relevant scales, the formalism was developed for Minkowski spacetime (still to be denoted by $M$), but the extension to curved geometries is nowadays completely clear. However, for the sake of clarity, we will present the idea underlying our approach in the flat case. The generalisation is (at least in principle) more or less straightforward. 

We already stressed that our point of view is substantially different from the ones based on non relativistic Quantum Mechanics. In the non relativistic case, in fact, the position of a single particle is an observable represented by a selfadjoint operator and the standard deviation $\Delta x^\mu$ is a good measure of the limits of our knowledge of it when the particle is in some specific state. For example, in Heisenberg's microscope experiment the operators $x_i$, with $i=1,2,3$ are position observables associated, say, to a specific electron. However, as we shall now briefly recall, this is \textit{not} the case in relativistic Quantum Field Theory (QFT).

A field is usually viewed as a way to attach physical degrees of freedom to a spacetime point. This means that the physical quantities the theory speaks about are precisely those degrees of freedom, while a point is to be seen as a mere label or (if coordinates are introduced) parameter. Its status may be compared with the one of time in non-relativistic Quantum Mechanics. This is why in QFT one (loosely) speaks about observables $O(x)$ at a point $x$ and we are not interested in limitations in the measurement of the position of some particle, but  in limitations  that the presence of gravity puts on localisability of quantum field theoretic observables. 

However, quantum fields do not really make sense when evaluated at a point. More precisely,
they are operator-valued tempered distributions: after ``smearing'' with some Schwartz test function $f$ (that is evaluated on $f$), they give (in general unbounded) operators acting on some Hilbert space $\cH$ with scalar product $(\cdot,\cdot)$ containing a distinguished vector $\psi_\omega$ called the vacuum. The Hilbert space $\cH$ is usually supposed to carry a unitary representation $U$ of the Poincar\'e group $\cP$. There is also mass operator $M=P^\mu P_\mu$, $\mu=0, \dots, 3$, where the $P_\mu$'s are the generators of translations. Moreover, there is a unique vacuum vector $\psi_\omega$, defined by $U(\Lambda)\psi_\omega=0$ for every $\Lambda\in\cP$. Conventionally (to fix ideas we will consider a real scalar field), we write
\begin{equation} 
\Phi(f)=\int_M \Phi(x) f(x) dx \qquad\qquad f\in\mathcal{S}.
\end{equation}
It is usually assumed that from the operators $\Phi(f)$, $f\in S$ with $\text{supp} f=\mathcal{O}$, we can construct all quantum observables localised in the region $\cal O$. They form an algebra $\cA(\cO)$ of operators that we label by $\cal O$ to indicate that they are localised there\footnote{For $\Phi$ a scalar field, the quantities $\Phi(f)$, $f\in S$, will themselves be observables. This is not true, for example, for a fermionic field.}.
Of course, the crucial ingredient of locality comes from the commutation relations of fields. From these one infers
\begin{equation}
[ \Phi(f_1),\Phi(f_2)] =0
\end{equation}
whenever $\text{supp} f_1$ and $\text{supp} f_2$ are space-like separated. The corresponding algebras then commute and observables with casually disjoint supports are compatible in the sense of Quantum Mechanics.\\
So far so good for what concerns observables. On the other hand, a \textit{state vector} $\psi\in\cH$ is said to be localised in some spacetime region $\cO$ if it ``looks like the vacuum'' for observations with support space-like separated from $\cO$, that is $(\psi,A\psi)=(\psi_\omega, A\psi_\omega)$ for \textit{any} $A\in\cA(\tilde{\cO})$ and $\tilde{\cO}$ space-like from $\cO$ \cite{knight}. Such vectors are generated from the vacuum by partial isometries in $\cO$, a typical example being $e^{i\Phi(f)}\psi_\omega$, $\text{supp} f\in \cO$ \cite{licht}.

To describe particles, one must first single out a single particle space, a subspace $\cH^{(1)}\subset\cH$. This is done by exploiting the assumed existence of the representation $U$ of $\cP$: it corresponds to the discrete part of the spectrum of the mass operator $M=P^\mu P_\mu$. Thus, vectors in $\cH^{(1)}$ represent single particle states. However the crucial fact, as first pointed out by Newton and Wigner \cite{wigner}, is that such vectors \textit{cannot} be interpreted as spacetime wave functions for the corresponding particle and this entails the well known problems in the definition of a position (or coordinate) observable in QFT.

Finally, we point out that Heisenberg's microscope experiment is described in the language of QFT by making use of the complicated machinery of scattering theory. Unfortunately, for the moment we deal with much more basic matters and as a consequence we absolutely do \textit{not} claim that the results concerning the localisation of particles obtained in non relativistic approaches are wrong, but only that at they are not relevant for the scopes of this paper. 

Consider now a sequence $f_n \to \delta_{p}$ (in a suitable sense). The spacetime supports $\cO_n$ of the operators $\Phi(f_n)$ shrink to the point $p$ with coordinates $x_p$ and so will the ones of the observables obtained from them. By using the partial isometries mentioned before, we obtain corresponding localized states. However, we have to pay a price: the average energy of the states will increase to infinity as the support shrinks to a point, as it can be seen by a simple application of Heisenberg's uncertainty relations (see below for more details). But if gravitational effects are disregarded there is in principle no higher bound whatsoever to the energy density of a state so that we can ideally approach sharp localisation of observable quantities at a point as much as we want, and in this way attach an operational meaning to the concept of point $p$ with coordinates $x_p$.

When gravity is taken into account, however, things change drastically. We are no more free to increase the energy density as much as we want due to the possible production of event horizons, which hide the region to a (distant) observer. This motivated in \cite{dop} the formulation of a Principle of gravitational stability against localisation of events (PGSL):
\begin{itemize}
 \item[] The gravitational field generated by the concentration of energy required
by the Heisenberg uncertainty principle to localise an event
in spacetime should not be so strong to hide the event itself to any
distant observer - distant compared to the Planck scale.
\end{itemize}
It should be clear that the terminology ``event'' is used above (and below) in the same sense
as Einstein did: an event is \textit{not} a physical process but rather a place (specified by some parameters called coordinates) where a physical process could take place.

To implement the PGSL, two ingredients are needed:
\begin{itemize}
 \item[1)] general conditions for the (non) formation of horizons on the relevant background;
 \item[2)] an estimate of the energy of the localised matter or radiation.
\end{itemize}
It was thus natural to begin with the minkowskian case, first described in \cite{dop}. There, a linear approximation of Einstein's equations was used as for $1)$ and the Heisenberg uncertainty relations for $2)$. We stress that the fact that the energy content of the localisation experiment was evaluated making use of \textit{non} relativistic Quantum Mechanics has nothing to do with the general quantum field theoretic point of view but must be regarded as a (certainly rough but surprisingly effective) approximation\footnote{In the case of a (three)-spherical localisation region $\cS$, it was possible to obtain mathematically rigorous estimates in a more realistic relativistic setting, that is describing matter by relativistic quantum fields \cite{dop2}. The results were completely similar.}.

All in all, they found
\begin{gather}\label{dfr}
c\Delta t\left(\Delta x^1+\Delta x^2+\Delta x^3\right)\geq {\lambda}^{2}_{P},\\
\Delta x^1\Delta x^2+\Delta x^1\Delta x^3+\Delta x^2\Delta x^3\geq {\lambda}^{2}_{P}.
\end{gather}
where the $\Delta x^\mu$'s, $\mu=0,\dots,3$ should be interpreted as lengths of the edges of the ``localising box''. The idea was to identify $\Delta x^\mu$, $\mu=0,\dots,3$ with the mean standard deviations
\begin{equation}
  \Delta x^\mu \doteq \Delta_{\omega_\phi} x^\mu = \sqrt{\omega_\phi((x^\mu)^2)-\omega_\phi(x^\mu)^2} = \sqrt{(\phi, (x^\mu)^2\phi)-(\phi, x^\mu \phi)^2},
 \end{equation}
of suitable hermitian operators $x^\mu$, $\mu=0,\dots,3$, acting on some Hilbert space $\cH$ with scalar product $(\cdot,\cdot)$ and satisfying commutation relations such that \eqref{dfr} follow (here $\omega_\phi$ is a so called vector state and $\phi\in \cH$). The $x^\mu$ are regarded to as (global) quantum coordinates on a quantum Minkowski spacetime $\cE$, which in turn was (as usual in non commutative geometry) identified with the $C^*$-algebra (see \cite{sakai}) they generate\footnote{To be more precise, an abstract analysis of \eqref{dfr} was performed in \cite{dop} to determine an \textit{abstract} $C^*$-algebra with generators $x^\mu$ having a (unique in this case) representation on some Hilbert space in which the $x^\mu$ are represented by hermitian operators and satisfy \eqref{dfr} for any state $\omega$ on $\cE$.}. This is of course in complete analogy with Quantum Mechanics, but with the crucial difference that here the same mathematical objects are interpreted in a completely different way. 

We stress that, coherently with a quantum field theoretic point of view, the $x^\mu$ should be seen as a space of non commutative \textit{parameters} that (together with the corresponding states on $\cE$) describe ``non commutative localisation properties'' of quantum fields now defined on $\cE$. These fields were actually constructed in \cite{dop} (but see also \cite{dop3},\cite{dop4}).

\subsection{Beyond the linear approximation: Penrose inequality}\label{secmink}

The fact that the linear approximation is unsatisfactory in deriving conditions that black holes do not form was of course clear since the beginning\footnote{Sergio Doplicher, private communication}. 
The main idea in \cite{TV}, which for future use we shortly recall below, was to replace it by the Hoop Conjecture, expressed in terms of Penrose's inequality (see \cite{Mars,S,B,F} and references therein): 
\begin{itemize}
 \item[] For asymptotically flat data, horizons form if and only if
\begin{equation}
A < 16\pi\frac{G^2}{c^4}M^2,
\label{1}
\end{equation}
where $A$ is the proper area enclosing the collapsing object of total mass $M$.
\end{itemize}
At this point, it is important to notice that the mass present in (\ref{1}) is nothing else the total mass (ADM mass) of the  black hole, i.e. the proper mass $M_p$ together with the (negative) contribution due to the binding gravitational energy, so that $M\leq M_p$. However, since to obtain our STUR we use the Heisenberg inequalities (\ref{Hei}), we need a formulation of the Penrose inequality in terms of a proper local mass within a finite region of proper area $A$. To this purpose, in the limiting case of asymptotically flat spacetimes with
spherically symmetric collapsing masses, 
we have at our disposal the theorems in \cite{Bi}, where the condition that horizons form is expressed only in terms of proper lengths and proper masses. Thanks to these theorems, 
we could be tempted to write the Penrose inequality (\ref{1}) in terms of $M_p$ to obtain a less restrictive and only necessary condition. Fortunately, we only need a sufficient condition to avoid horizon formation in a thought experiment, as provided by inverting (\ref{1}) and setting $M=M_p$. In the following, we drop the subscript ``$p$'' for masses and energy enclosed within $A$.
To summarise, we obtain the following \textit{sufficient} condition for no black hole formation as:
\begin{equation}
A\geq 16\pi \frac{G^2}{c^4}M^2.
\label{2}
\end{equation}
We need to evaluate the proper area $A$ of the localising region. As a working approximation and following \cite{TV}, we take the one of the background (that is the chosen spacetime \textit{without} the experiment). Choosing cartesian coordinates and indicating by $\Delta x_i$, $i=1,2,3$, the sides' lengths of a parallelepiped, we may write
\begin{equation}
A = \frac{\Delta A}{\beta^2}=\frac{\Delta x^1\Delta x^2+\Delta x^1\Delta x^3+\Delta x^2\Delta x^3}{\beta^2}, \label{deltaA}
\end{equation}
in view of the fact that our background is now by assumption minkowskian.
Note that in order to ensure that in the spherical case ($\Delta x^1=\Delta x^2=\Delta  x^3 =2\Delta R$ with obvious notation) we have $A=4\pi{\Delta R}^2$, we need to take $\beta^2=3/(\pi)$ \cite{sabine}.

We estimate the energy of the collapsing field configuration making use of Heisenberg's uncertainty relations
\begin{equation}\label{Hei}
\Delta x_j\Delta p_j \geq \frac{\hbar}{2},\;i,j=1,2,3 \qquad\qquad \Delta t\Delta E\geq \frac{\hbar}{2},
\end{equation}
(here $p$ is the momentum, $t$ the time, $E$ the energy). For a single particle we have $E^2\geq c^2 (p_1^2+ p_2^2+  p_3^2)$, so that substituting in (\ref{2}) we arrive at the relation:
\begin{equation}\label{ssmink}
\Delta A \geq 12 \lambda^4_P
\left[{\left(\frac{1}{\Delta x^1}\right)}^2+{\left(\frac{1}{\Delta x^2}\right)}^2+{\left(\frac{1}{\Delta x^3}\right)}^2\right].
\end{equation}

To express the preceding inequality in terms of geometrical quantities only, we observe that
\begin{equation}\label{schw}
\frac{\Delta A}{\sqrt{3}\Delta V} \leq \sqrt{{\left(\frac{1}{\Delta x^1}\right)}^2+{\left(\frac{1}{\Delta x^2}\right)}^2+{\left(\frac{1}{\Delta x^3}\right)}^2} \leq \frac{\Delta A}{\Delta V},
\end{equation}
where $\Delta V=\Delta x^1 \Delta x^2 \Delta x^3$ stands for the three-volume of the localising box. This shows that the ratio $\Delta A /\Delta V$ has exactly the same behaviour than the right hand side of \eqref{ssmink} as a function of $\Delta x^i$, $i=1,2,3$. To get a sufficient condition we use the right hand side inequality, so that the final form of our minkowskian space-space STUR is
\begin{equation}
{(\Delta_\omega x^1)}^2 {(\Delta_\omega x^2 )}^2 {(\Delta_\omega x^3 )}^2 \geq 12
\lambda^{4}_{P}
\left(\Delta_\omega x^1\Delta_\omega x^2+\Delta_\omega x^1\Delta_\omega x^3+\Delta_\omega x^2\Delta_\omega x^3\right).\label{mink1}
\end{equation}
while for our minkowskian time-space STUR we obtain (this time we use the time-energy Heisenberg relation)
\begin{equation}
{(c\Delta_\omega t)}^2\left(\Delta_\omega x^1\Delta_\omega x^2+\Delta_\omega x^1\Delta_\omega x^3+\Delta_\omega x^2\Delta_\omega x^3\right) \geq 12
\lambda^{4}_{P}.\label{mink2}
\end{equation}
Here, we added the subscript $\omega$ to emphasise the dependence on states now that the inequalities have a quantum interpretation. For later reference, we mention that
\begin{equation}
\Delta_\omega x^1\Delta_\omega x^2\Delta_\omega x^3 
\geq \sqrt[4]{12^3} \lambda_P^{3},
\label{36}
\end{equation}
must then hold \cite{TV}, so that a lower bound for the volume appears. From it, the existence of a mean maximal mass-energy density for such states can be deduced.

Significantly, the DFR uncertainty relations could be recovered as a weakening\footnote{In particular, the ratio $\Delta A /\Delta V$ is substituted by the smaller quantity $(\Delta A)^{-1/2}$. Thus, localisation regions with zero volume are allowed.} of \eqref{mink1},\eqref{mink2} by making use of the following simple algebraic inequalities
\begin{gather}
(a+b+c)^2\geq ab+bc+ac, \label{algebraic}\\
(ab+bc+ac)^3\geq a^2 b^2 c^2.\nonumber
\end{gather}
It is thus reasonable to regard the DFR commutation relations as a reasonable approximation to commutation relations implementing \eqref{mink1},\eqref{mink2}. In Section \ref{sec:quantum}, we will use the same approximation to build a DFR-like quantum Friedmann 
expanding spacetime.

\section{Isoperimetric inequalities in Friedmann-Walker cosmological backgrounds}

In order to extend the reasoning presented in section \ref{secmink} to cosmological spacetimes, we face the problem that condition (\ref{2}) does not hold in general non asymptotically flat spacetimes. As it is well known, establishing mathematically rigorous conditions about the formation of horizons is a formidable task. However, if we restrict our attention to trapped surfaces, then inequalities are at our disposal in the form of exact theorems that provide necessary or sufficient conditions that they do not form. 

A first class of results (see \cite{Br}) 
concerns spherically symmetric collapses in open spatially flat universes 
under the assumption that initial data are given such that the trace of the extrinsic curvature is spatially constant. 
They assert that a spatial spherical surface $S$ with proper radius
\begin{equation}
L_S=\int_{0}^{R}a(t)dr=a(t)R,\,\,A =4\pi a(t)^2R^2,\label{4}
\end{equation}
and area $A$ is not trapped if and only if the amount of the 
positive excess proper mass $\delta M$ 
(the extra mass apart from the constant background density filling the spacetime) inside $S$ satisfies
\begin{equation}
 \frac{L_S}{2}+\frac{HA}{4\pi c} > \frac{G}{c^2}\delta M,\label{3}
\end{equation}
The excess proper mass $\delta M_S$ must be positive, i.e. the weak energy condition holds. The proper quantities $L_S$ and $A$ both refer to the 
spacetime with the excess mass $\delta M$, i.e. the backreaction is taken into account.
The condition approaches Penrose's inequality (\ref{2}) for the spherical case as we take smaller values of $H$\footnote{It is worth noticing that the STUR above can be extended to a de Sitter spacetime: the only change to be made is to set $H(t)=c\sqrt{\Lambda/3}$.}.\\
Moreover, for spacetimes conformally equivalent to Friedmann flat solutions (\textit{i.e.} with metric $\hat{g}_{\mu\nu}=\Phi^4 g_{\mu\nu}$ where $g_{\mu\nu}$ is Friedmann)
a theorem in \cite{Mal} gives a necessary condition that a (generic, not necessarily spherical) equipotential spatial surface\footnote{This is a surface on which $\Phi$ is constant. It can be shown that outside matter this is equivalent to constancy of constant redshift.} $S$ with proper area $A$ be trapped:
\begin{equation}
\frac{G}{c^2}\delta M \geq \frac{1}{2}\sqrt{\frac{A}{\pi}}+A\sqrt{\frac{\rho}{6\pi}},
 \end{equation}
with $\delta M$ the excess proper mass and $\rho = 3H^2/8\pi G$. Inverting the inequality, we get the desired sufficient condition that black holes do not form.

Summing up, it is reasonable to consider the following generalisation of Penrose's isoperimetric inequality for Friedmann flat expanding cosmologies: {\it black holes do not form if the (positive) excess of proper mass} $\delta M$ {\it inside a surface $S$ of proper area $A$ satisfies the inequality}:
\begin{equation}
\frac{\sqrt{A}}{4\sqrt{\pi}}+\frac{HA}{4\pi c} \geq \frac{G}{c^2}\delta M.\label{7}
\end{equation}
In the next section we apply the former generalised isoperimetric conjecture to obtain new STUR.

\section{STUR in Friedmann flat spacetimes: general case}\label{sec:general1}

We can take advantage of inequality (\ref{7}) to obtain new STUR. Suppose then we want to enclose a certain amount $\delta E=c^2 \delta M$ of energy during a fixed time interval $\Delta t$ (here and in the following $t$ will indicate \textit{proper} time) in a certain space box $S$ of proper area $A$, with the condition that no horizons form during the experiment. Then, we see from the preceding section that it is reasonable to impose
\begin{equation}
\frac{\sqrt{ A}}{4\sqrt{\pi}}+\frac{\langle H\rangle A}{4\pi c} \geq \frac{G}{c^4}\delta E.\label{20}
\end{equation}
where by $\langle H\rangle$ we indicate the mean value of $H$ during $\Delta t$.\\
First, to evaluate $\delta E$ we shall use the quasi-local energy (see \cite{TV,H,Sza}). Second, we stress that according to \eqref{7} only proper distances should appear in \eqref{20}. However, in the following we will estimate the proper area of the localising box.\\ 
With this in mind, we introduce a tetrad frame  with spatial axes along the sides of the box. 
Such an experimenter measures the proper time $t$ and the (infinitesimal) proper length at a given fixed time $t$ in terms of ${\eta}^{(a)}$:
\begin{equation}
d{\eta}^{(a)}= e^{(a)}_{\mu}dx^{\mu}, \label{210}
\end{equation}
where ${x^{\mu}}=(t,x^i)$ and $(a)=0,1,2,3$. From now on, we shall drop any specific notation for tetrad indices.
Since we want to test the non commutative structure at Planck's length, it is natural to suppose that the ``proper'' localising box has comparable dimensions. Thus, to determine the proper lengths of the spatial sides of the box, an integration of \eqref{210} at a given fixed time $t$ and over a small spatial domain suffices. This gives
$\Delta \eta^i = \eta^i_2 - \eta^i_1$,
where of course $\eta^i_k$, $k=1,2$, indicate the 'proper'
coordinates of the two extremes of the relevant edge and, thanks to \eqref{210},
\begin{equation}
{\eta}^i= a(t) x^i.\label{21}
\end{equation}
Furthermore, since at a fixed arbitrary $t$ the metric is spatially flat, the area of the box has the same form of the minkowskian one, albeit expressed in terms of the proper length of the edges $\Delta{\eta}^{(a)}$. Moreover, since in a Friedmann universe the photon energy is $E=hc/\lambda a(t)$ with $\lambda a(t)$ the photon
proper wavelength, a comoving observer is legitimated to write the Heisenberg uncertainty relations as follows:
\begin{gather}
\Delta_\omega E\Delta_\omega t\geq\frac{\hbar}{2},\qquad
\Delta_\omega p_{{\eta}^i}\Delta_\omega{\eta}^i\geq\frac{\hbar}{2},\;\;\forall i=1,2,3.\label{23}
\end{gather}
for any state $\omega$.
The derivation of the STUR in a generic stationary asymptotically flat background now parallels the procedure in Section 2, with the only proviso that we must rewrite everything in terms of the quantities $\Delta_\omega \eta^i$, $i=1,2,3$ and $\Delta_\omega t$\footnote{We remind the reader that the symbols $\Delta_\omega A$ or $\Delta_\omega V$ are shorthands for the corresponding combinations of these uncertainties and do \textit{not} indicate standard deviations of area or volume \textit{operators}.}. In particular,
we derive the space-space STUR using  inequalities
\eqref{schw} to estimate the right hand side of \eqref{20} in terms of the ratio $\Delta_\omega A/\Delta_\omega V$. 
Thus, the final form of our space-space STUR is
\begin{equation}
\frac{\sqrt{\Delta_\omega A}}{4\sqrt{3}}+\frac{\omega(H)\Delta_\omega A}{12c}\geq \frac{\lambda^2_P}{2} \frac{\Delta_\omega A}{\Delta_\omega V}.\label{27}
\end{equation}
For the time-space STUR, thanks to (\ref{20}) and (\ref{23}) we get:
\begin{equation}
c\Delta_\omega t\left(\frac{\sqrt{\Delta_\omega A}}{4\sqrt{3}}+\frac{\omega(H)\Delta_\omega A}{12c}\right)\geq \frac{\lambda^2_P}{2}.\label{28}
\end{equation}
Expressions (\ref{27}), (\ref{28}) represent the generalisation of the minkowskian STUR in \cite{TV} to Friedmann flat spacetimes.\\
It is important to note that under very reasonable hypothesis $\omega(H) >0$ and the presence of a non zero Hubble flow reflects itself in {\it weaker} uncertainty relations than in the flat case. So to put it, since making black holes in expanding spacetimes is harder then localisation is easier. One may think of this fact in terms of the appearance of an effective (smaller, $H$-dependent) Planck length. Correspondingly, we may say that the minkowskian STUR are \textit{stronger} the Friedmann-like ones: the inequalities are stricter. We conclude that should we consider analogous less stringent inequalities then Friedmann's, then we would be departing even more from the flat case. In Section \ref{sec:quantum}, when presenting a concrete realisation of a Friedmann-like non commutative spacetime, we will actually start from a more workable less stringent form of (\ref{27}), (\ref{28}).

\section{Big bang Friedmann flat cosmologies}\label{sec:general2}

In this Section we shall study some general properties of states ``far from'' the big bang singularity which can be drown from the STUR \eqref{27},\eqref{28}. We will work under the following assumptions:
\begin{itemize}
 \item[a)] the existence of a quantum (non commutative) Friedmann spacetime, \textit{i.e.} a ($C^*$)-algebra $\cal E$ satisfying certain desirable properties as for example specified in subsection \ref{sub:local} and Section \ref{sec:quantum}. The STUR are fulfilled with respect to all states of $\cal E$. Here and in the following any expressions involving the evaluation of states on (generalised, unbounded) elements of $\cal E$ will be assumed to make sense and the
 self adjoint generator $t$ to have a positive spectrum 
 (\textit{i.e.} $sp(t)\in{{R}}^{+}$).
 \footnote{This is motivated by the requirement of a reasonable classical limit, see section \ref{sec:quantum}.} 
 \item[b)] All states are considered to have geometrical meaning (see below).
\end{itemize}
Condition $b)$ is motivated by the fact that, as customary in the case of constrained quantum systems, not all states need to be considered. For example, in noncommutative geometry only \textit{pure} states (\textit{i.e.} states that are not mixtures) are identified with points. Moreover, in \cite{dop} among all pure states only those having a certain property of ``maximal localisation'' are considered. We emphasise that this is a completely different restriction than the one introduced in Subsection \ref{sec:horizon}, where in the presence of particle horizons we will introduce one more condition to select, among the geometrically meaningful ones, states corresponding to physically realisable localisation experiments.

\subsection{Away from the big bang}

To start with, we need a quantitative characterisation of what we mean for a state $\omega$ to be ``away from the big bang''. To this end, we observe that on the one side the quantity $1/H(t)$ is classically an approximation of the age of the average universe, and on the other it is reasonable to ask for the preparation of such a state (localisation experiment) not to have started at ages near the big bang. Thus, we might wish to impose the condition that light take a smaller time than $\omega(H^{-1})$ to travel the whole diameter of the localisation region itself. This means asking that $\sqrt{\Delta_\omega A} << c \omega(H^{-1})$. However, for the sake of simplicity, we will make use of the inequality $1\leq \omega(a)\omega(a^{-1})$ (valid for positive invertible operators $a$) and impose the thus stronger requirement that
\begin{equation}
\omega(H)\sqrt{\Delta_\omega A} << c.
\label{away}
\end{equation}
To see what this means, consider a box of size $\sim 10^{35}\lambda_P \simeq 1$ meter. Condition
(\ref{away}) implies that $\omega(H) << 10^{8}\; s^{-1}$. As a result, for a macroscopic box of this size also the beginning of the nucleosynthesis ($t\sim 10^{-2}\;s$) can be considered ``away from the big bang''. However, taking
a box of size $\sim {\lambda}_P$ and assuming $\omega(H)\simeq \omega(t)^{-1}$, condition (\ref{away}) takes the form $\omega(t)>> t_P$.\\
Using (\ref{away}) to eliminate the $H$ dependence in (\ref{27}),(\ref{28}), for states $\omega$ away from the big bang we see that the inequalities
\begin{gather*}
{(c\Delta_\omega t)}^2\left(\Delta_\omega\eta^1\Delta_\omega \eta^2+\Delta_\omega \eta^1\Delta_\omega \eta^3+\Delta_\omega \eta^2\Delta_\omega \eta^3\right) \geq u_1
\lambda^{4}_{P},\\
{(\Delta_\omega \eta^1)}^2 {(\Delta_\omega \eta^2 )}^2 {(\Delta_\omega \eta^3 )}^2 \geq u_2
\lambda^{4}_{P}
\left(\Delta_\omega \eta^1\Delta_\omega \eta^2+\Delta_\omega \eta^1\Delta_\omega \eta^3+\Delta_\omega \eta^2\Delta_\omega \eta^3\right),
\end{gather*}
with some constants $u_1, u_2>0$ must be automatically satisfied. They are completely analogous to the minkowskian case, so that as in Subsection \ref{secmink} we have
\begin{equation*}
\Delta_\omega \eta^1\Delta_\omega \eta^2 \Delta_\omega \eta^3 \geq u_2^{3/4} \lambda^{3}_{P}.
\end{equation*}

\subsection{Cosmologies with particle horizons}\label{sec:horizon}

We now restrict our attention to power law cosmologies with particle horizons, where  $a(t)=t^{\alpha},\;\alpha\in(0,1)$ and $H(t)=\alpha/t$. Classically, for an observer at (proper) time $t$ the particle horizon is the maximal proper spatial region that is in causal relation with him. It is a spherical region with proper radius 
\begin{equation}
{\eta}_h(t)= a(t)\int_{0}^{\bar t}\frac{cdt^{\prime}}{a(t^{\prime})},\label{H1}
\end{equation}
so that
\begin{equation}
{\eta}_h(t)=\frac{ct}{(1-\alpha)}=\frac{c\alpha}{(1-\alpha)H(t)}.
\label{H2}
\end{equation}
As a result, an experimenter at time $t$ cannot localise spatial regions with radius bigger than \eqref{H2}. Points outside the horizon certainly have a geometrical meaning, but since are completely inaccessible to the given observer have no physical meaning at all for him. For this reason, when moving to the quantum level we will consider a geometrically meaningful state to be physical whenever
\begin{equation}
\max_{i=1,2,3}\{ \omega(H)\Delta_\omega {\eta}^i\}\leq\frac{\alpha c}{(1-\alpha)},\label{H2a}
\end{equation}
so that we obtain
\begin{equation}
\omega(H) \sqrt{\Delta_\omega A}\leq \sqrt{3}\frac{\alpha c}{(1-\alpha)}.\label{H3}
\end{equation}
Combining (\ref{H3}) with (\ref{27}), (\ref{28}) we obtain
\begin{gather*}
{(c\Delta_\omega t)}^2\left(\Delta_\omega\eta^1\Delta_\omega \eta^2+\Delta_\omega \eta^1\Delta_\omega \eta^3+\Delta_\omega \eta^2\Delta_\omega \eta^3\right) \geq 12
\tilde{\lambda}_P^{4},\\
{(\Delta_\omega \eta^1)}^2 {(\Delta_\omega \eta^2 )}^2 {(\Delta_\omega \eta^3 )}^2 \geq 12
\tilde{\lambda}_P^{4}
\left(\Delta_\omega \eta^1\Delta_\omega \eta^2+\Delta_\omega \eta^1\Delta_\omega \eta^3+\Delta_\omega \eta^2\Delta_\omega \eta^3\right),
\end{gather*}
where
\begin{equation}
{\tilde{\lambda}}_P=\lambda_P\sqrt{1-\alpha}.\label{H6}
\end{equation}
Thus, inequalities formally identical to the minkowskian STUR are satisfied, albeit with an ``effective'' Planck's length\footnote{Note that $\tilde{\lambda}_P$ becomes smaller when approaching $\alpha=1$, the limit value for inflationary cosmologies. Thus, for these cosmological solutions of Einstein's equations the effective length scale can be smaller than $\lambda_P$. In any case cosmologies with 
$\alpha\in(2/3,1)$ have negative hydrostatic pressure $P$ and then do not represent usual matter.} ${\tilde{\lambda}}$, with all consequences. In particular, the existence of a minimal volume is especially remarkable, because it indicates that if the expansion is sufficiently fast at early times the restriction to physical states may rule out a big bang pointlike singularity.

We now provide some more evidence towards this conclusion. Although as expected these imply no lower limit on the precision of the measurement of a \textit{single} (proper) coordinate, it is not difficult to see that the space-space minkowski like inequality provides
\begin{equation}
 \max_{i=1,2,3}\{\Delta_\omega \eta^i\}\geq \sqrt{6}{\tilde{\lambda}}_P. \label{max}
\end{equation}
But then from \eqref{H2a} we immediately deduce 
\begin{equation} 
\omega(H)\;\leq\frac{\alpha}{(1-\alpha)\sqrt{6(1-\alpha)}t_P}.
\label{H9}
\end{equation}
Now, taking into account the classical Einstein's equations $H^2=8\pi G\rho/3$, we infer
\begin{equation}
\omega(\rho)_\textrm{max}=\frac{\alpha^2}{16\pi{(1-\alpha)}^3}{\rho}_P,
\label{dens}
\end{equation}
where $\rho$ indicates the mass density. This equality is consistent with well known results from loop quantum gravity (\cite{Ast} and more recently \cite{vid}) and can be regarded at as an additional indication that quantum mechanical effects prevent the appearance of a big bang singularity (when particle horizons are present).

\section{Inflationary models without big bang}

In section 5, we studied our STUR in spacetimes under the condition that $t$ be positive.
This condition is certainly suitable for cosmological models starting at $t=0$ with the big bang. 
However, for completeness
we can consider another class of models motivated within the  inflationary paradigm \cite{3a}-\cite{9a},
with $t$ ranging in the whole real line $\bR$ and a regular metric (no big bang) at $t=0$.\\ 
As an example, we can consider the proposal advanced in \cite{infl2}, where an over horizon universe is represented by a De Sitter metric with embedded spherical bubble of true vacuum evolving separately as Friedmann universes. 
As an extreme idealisation, it should be possible in principle to build an early inflationary
universe emerging from a Minkowski spacetime (see for example \cite{infl1}).\\
In any case, what it is important for our purposes is that in the usual approach the curvature of the spacetime is neglected,
since during the early inflationary epoch the universe expanded very quickly. Hence, the primordial inflation can be well depicted 
in terms of a spatially flat metric. For the description of inflation,
one introduces a potential
$V(\phi)$, where the scalar potential $\phi$ is a time function satisfying the Klein Gordon equation in a flat background:
\begin{equation}
\frac{{\phi}_{,t,t}}{c^2}+3H\frac{{\phi}_{,t}}{c^2}+\frac{dV}{d\phi}=0.
\label{infl}
\end{equation}
Note that the potential $V(\phi)$, after solving equation (\ref{infl}) in terms of  $\phi(t)$,
can also be expressed directly in terms of the cosmological time $t$  \cite{cher1,vinfl}.
This fact is important, since it allows to extend the construction of the next section 
to this case.
Moreover, as it is well known the introduction of a generic potential $V(\phi(t))$ is equivalent to adding in the Friedmann equations the contribution of a perfect fluid with 
pressure $p_{\phi}$ and energy density ${\rho}_{\phi}$ and equation of state $p_{\phi}=c^2\gamma(t)\rho_{\phi}$. The result is
\begin{gather}
H^2(t)=\frac{8}{3}\pi G\left(\rho+{\rho}_{\phi}\right),\label{I1}\\
{\rho}_{\phi}=\frac{{\phi}_{,t}^2}{2c^4}+\frac{V(\phi)}{c^2},\;\;\;p_{\phi}=\frac{{\phi}_{,t}^2}{2c^2}-V(\phi),\label{I2}
\end{gather}
where the density $\rho$ denotes other kinds of usual matter (radiation,  dark matter...) already
present in the universe.
As a consequence, even for these cosmological models the STUR (\ref{27}) and (\ref{28})
hold unchanged.
The solutions of (\ref{I1}) and (\ref{I2})  obviously lead to
Friedmann-Lemaitre cosmologies: 
they are spatially homogeneous and 
the dimension $s$ of the group of transitivity is $s=3$ and the one of the  group of isotropy is  $g=3$. 
Hence, such spacetimes always admit six Killing vectors, with the exception
of de Sitter spacetime ($s=4$ and $g=6$). This is the second ingredient needed to include all such spacetimes in the class of quantum models
constructed in the next section.

\section{Models of quantum expanding spacetimes}\label{sec:quantum}

In this section 
we outline the construction of concrete models of quantum Friedmann cosmological spacetime and then briefly discuss some of their properties. Notice that an assignment of a classical background is equivalent to fixing a specific Hubble flow $H$. Correspondingly, we will build models quantum Friedmann spacetime regarding $H$ as the given input. However, we will not take the STUR (\ref{27}), (\ref{28}) as our starting point but rather the following weaker but more workable version:
\begin{gather}\label{stur2}
\Delta_\omega A\left(\frac{1}{4\sqrt{3}}+\frac{\omega(H)\sqrt{\Delta_\omega A}}{12 c}\right)\geq \frac{\lambda^2_P}{2},\\
c\Delta_\omega t\left( \Delta_\omega \eta_1 +\Delta_\omega \eta_2 +\Delta_\omega \eta_3\right)\left(\frac{1}{4\sqrt{3}} +\frac{\omega(H)\sqrt{\Delta_\omega A}}{12 c}\right)\geq \frac{\lambda^2_P}{2},\label{stur3}
\end{gather}
where of course $\Delta_\omega A$ is just a shorthand as in \eqref{deltaA}. It is obtained by making use of inequality \eqref{algebraic} to eliminate the dependence on the three-volume $\Delta_\omega V$ in (\ref{27}) and substituting $\sqrt{\Delta A}\leq \sum \Delta \eta_i$ in (\ref{28}).

For the sake of definiteness, we collect below some desirable properties of a ``concrete quantum model''. Notice the reference to a specific Hilbert space: we will not consider the problem of defining an \textit{abstract} $C^*$-algebra whose inequivalent representations provide all (reasonable) ``concrete quantum models''.

{\bf Definition.} A (complex) *-algebra $\cal A$ of operators acting on some (separable) Hilbert space $\cH$ with (self adjoint unbounded) generators $\eta_\mu$, $\mu=0,\dots,3$, is said to be a concrete covariant realisation of the non commutative spacetime $\cal M$ corresponding to the (classical) spacetime $M$ with global isometry group $G$ if:
\begin{itemize}
 \item[1)] the relevant STUR are satisfied;
 \item[2)] there is a (strongly continuous) unitary representation of the global isometry group $G$ under which the operators $\eta$ transform as their classical counterparts (covariance);
 \item[3)] there is some reasonable classical limit procedure for $\lambda_P \to 0$ such that the $\eta_\mu$ become in an appropriate sense commutative coordinates on some space containing the manifold $M$ as a factor.
\end{itemize}

In what follows, we will focus on constructing quantum coordinate operators on some separable (\textit{i.e.} having a countable basis) Hilbert $\cH$ space having the properties stated in 1), 2) above. We recall that all flat Friedmann metrics share one important characteristic: six Killing vectors generating a $SO(3)\ltimes \bR^3$ global isometry group. Concerning the only exception, de Sitter spacetimes, one can show that our STUR still hold and a slight modification of the construction below provides operators with the correct commutation relations. However, in this case we do not have a ``quantum model'' since we cannot implement the whole (nonlinear) isometry group.

We leave for the future a more detailed study of crucial mathematical properties such as the existence of suitable invariant domains and so on. We expect them to be essential for the understanding of the representation theory of the ($C^*$)-algebra generated by the coordinate operators. Also, we will not try to seriously address item 3) but only ask that the spectrum sp$(\eta_k)$ of the quantum coordinate operators $\eta_k$, $k=1,2,3$ be always the real numbers $\bR$ and the spectrum sp$(t)$ of $t$ be the $\bR$ or just the strictly positive reals $\bR^+$.

A remark is here in order. In Quantum Mechanics what matters is not the concrete Hilbert space on which the observable operators act, but their spectrum and the abstract algebra they generate. To be more concrete, since for any two separable Hilbert spaces $\cH_1,\cH_2$ there is a unitary $W:\cH_1\to\cH_2$, then the physics provided (and the algebra generated) by any (bounded) operators $S,T$ on $\cH_1$ and their commutator $[A,B]$ will be the same as the one provided by $WAW^{-1}, WBW^{-1}$ and $W[A,B]W^{-1}=[WAW^{-1},WBW^{-1}]$.

For the sake of clarity we start by a lemma which motivates all the construction below. Roughly speaking, it says that if we find operators on some Hilbert space satisfying a suitably modified version of the minkowskian STUR (in which the Planck length so to say depends on $H(t)$) then we know that they satisfy \eqref{stur2},\eqref{stur3}.

\begin{Lemma}\label{fundamental}
Consider the equation $w^3 + k w^2 - 2\sqrt[4]{3}k = 0$ ($k= \sqrt{4c/3\lambda_P H}$, $H\geq 0$) and its solution $f_0(H)\colon \bR^+\to\bR$. Suppose the selfadjoint operators $t,\eta_k$, $k=1,2,3$, on the Hilbert space $\cH$ satisfy the uncertainty relations
\begin{gather}
 c\Delta_\omega t\left(\Delta_\omega{\eta}^1+\Delta_\omega{\eta}^2+\Delta_\omega{\eta}^3\right) \geq 
\lambda^{2}_{P}\abs{\omega\left(f_0\circ H(t)\right)}, \label{eq:general1}\\
\Delta_\omega{\eta}^1\Delta_\omega{\eta}^2+\Delta_\omega{\eta}^1\Delta_\omega{\eta}^3+
\Delta_\omega{\eta}^2\Delta_\omega{\eta}^3\geq 
\lambda^{2}_{P}\abs{\omega\left(f_0\circ H(t)\right)} \label{eq:general2}.
\end{gather}
Then they also satisfy modified STUR \eqref{stur2}, \eqref{stur3}.
 \end{Lemma}

\begin{proof}
To start with, observe that $f_0\to\sqrt[8]{12}$ as $H\to0$ and $f_0\sim(c/H\lambda_P)^{2/3}$ as $H\to +\infty$. Moreover, $f_0 > 0$ for $H\geq 0$ and hence $f_0\circ H(t)$ is a positive (bounded) operator. Plugging inequalities \eqref{eq:general1} and \eqref{eq:general2} into the left hand sides of \eqref{stur2} and \eqref{stur3} respectively we obtain the smaller quantity
\begin{gather*}
\cF_\omega =\abs{\omega(f_0)}\!\left(\frac{1}{4\sqrt{3}} +\frac{\omega(H)\sqrt{\abs{\omega(f_0)}}\lambda_P}{12c}\right).
\end{gather*}
We now use a consequence of Schwartz inequality $\omega(a^* a)\omega(b^* b) \geq \omega(ab)^2$ for positive self-adjoint operators $a$, that is $\sqrt{\omega(a)}\geq \omega(\sqrt{a})$:
\begin{gather*}
\cF_\omega\geq \omega(f_0)\left(\frac{1}{4\sqrt{3}}+\frac{\omega(H)\omega(f^{1/2}_0)}{12c/\lambda_P}\right) \geq \omega(f_0)\left(\frac{1}{4\sqrt{3}}+\frac{\omega(H^{1/2}f^{1/4}_0)^2}{12c/\lambda_P}\right)\geq \\
\geq \frac{\omega(f^{1/2}_0)^2}{2}\omega\left(\frac{1}{2\sqrt[4]{3}}+\frac{H^{1/2}f^{1/4}_0}{2\sqrt{3c/\lambda_P}}\right)^2
\geq \frac{1}{2}\omega\left(f^{1/4}_0\sqrt{\frac{1}{2\sqrt[4]{3}}+\frac{H^{1/2}f^{1/4}_0}{2\sqrt{3c/\lambda_P}}}\right)^4.
\end{gather*}
It is now obvious that the  STUR \eqref{stur2},\eqref{stur3} will be satisfied as long as
\begin{gather*}
\omega\left(f^{1/4}_0\sqrt{\frac{1}{2\sqrt[4]{3}}+\frac{H^{1/2}f^{1/4}_0}{2\sqrt{3c/\lambda_P}}}\right)^4= 1.
\end{gather*}
is satisfied for any state $\omega$. But this is a consequence of the operator equality
\begin{equation}\label{eq:uations}
f^{1/2}_0\left(\frac{1}{2\sqrt[4]{3}}+\frac{H^{1/2}f^{1/4}_0}{2\sqrt{3c/\lambda_P}}\right)= I,
\end{equation}
which by functional calculus follows from the cubic equation in the statement.\\
\end{proof}

We now come to the main result of this section.

\begin{Proposition}\label{prop:frie1}
Let $M=I\times \Sigma$ be the background manifold, $t$ the universal time, $H(t)\colon I\to \bR^+$ be the associated Hubble parameter and $f_0\colon \bR^+\to \bR$ be the bounded smooth strictly positive function of Lemma \ref{fundamental}. Then there are operators $\eta_\mu$, $\mu=0,\dots,3$ on the Hilbert space $L^2(\bR^5, d\xi)\otimes L^2(SO(3),d\mu(R))$\footnote{Here and in what follows $d\mu(R)$ indicates the unique invariant Haar measure and $R\in SO(3)$.}, essentially selfadjoint on a domain $\cD$ to be specified below, such that for any state $\omega$ in their domain and any sufficiently regular strictly positive real function $f$ on $\bR^+$ the inequalities \eqref{stur2} and \eqref{stur3} hold true. Moreover, there is a unitary representation $U$ of the group $G=SO(3)\ltimes\bR^3$ such that
\begin{equation*}
U(R,a)\eta_i U(R,a)^{-1}=R_{ik} \eta_k + a_i I,\quad U(R,a)\eta_0 U(R,a)^{-1}=\eta_0,
\end{equation*}
for $(R,a)\in SO(3)\ltimes\bR^3$.
\end{Proposition}

\begin{proof}
 Thanks to Lemma \ref{fundamental}, we only have to construct operators on some Hilbert space satisfying inequalities \eqref{eq:general1}, \eqref{eq:general2} \textit{and} a unitary strongly continuous representation $U$ of the isometry group $G=SO(3)$. Consider now the operators $(\tilde{t},\tilde{\eta}_i)$, $i=1,2,3$, on $L^2(\bR^5, d\xi)$ defined in Lemma \ref{lem:frie} of Appendix 1 and fix some $\gamma$ such that Im$(\gamma)=I$. Then, the first condition is perfecly met but we lack $U$. Following \cite{Pia2}, we thus proceed with a direct integral construction on the group $G$ and define our coordinate operators $(t=\eta_0,\eta_i)$, $i=1,2,3$ on $L^2(\bR^5, d\xi)\otimes L^2(SO(3),d\mu(R))$ setting
\begin{gather}
 t\left(\phi_1(\xi)\otimes \phi_2 (R)\right) =\tilde{t}\phi_1(\xi)\otimes \phi_2(R), \\
 \eta_i\left(\phi_1(\xi)\otimes \phi_2 (R)\right) =R_{ij}\tilde{\eta}_j\phi_1(\xi)\otimes \phi_2(R) = \tilde{\eta}_j\phi_1(\xi)\otimes \phi_2(R)R_{ij}
 \end{gather}
where $i,j=1,2,3$ and $R_{ij}=R_{ij}(R)$ indicates the matrix corresponding to $R$ in the defining three dimensional representation of $SO(3)$. Their commutation relations read
\begin{gather}\label{eq:noncommcomm}
[\eta_\mu,\eta_\nu ] = i\sqrt{3}\lambda_P^2 \cK_{\mu\nu}\qquad \cK_{\mu\nu}\left( \phi_1\otimes \phi_2\right) = R_{\mu\rho} R_{\nu\sigma}\tilde{A}^{\rho\sigma}\phi_1\otimes \phi_2, \\
R_{\nu\sigma}=\text{diag}\{1,R_{ij}\}, \;\; \tilde{A}_{\rho\sigma}= \{[\tilde{\eta}_\rho,\tilde{\eta}_\sigma ]\},\;\; \mu,\nu,\rho,\sigma=0,\cdots,3.\nonumber
\end{gather}
By construction we now get the unitary strongly continuous $SO(3)$-action
$$
U(R^\prime)\left(\phi_1(\xi)\otimes \phi_2(R)\right)=\phi_1(\xi)\otimes\phi_2(R^\prime R).
$$
We now define new unitaries $U(a)=\tilde{U}(a)\otimes I$ for $a\in\bR^3$, where $\tilde{U}$ is the representation of $\bR^3$ defined in Lemma \ref{lem:frie}. Then a direct calculation shows that the unitaries $U(R,a)= U(a)U(R)$ are such that
\begin{gather}\label{eq:friedcovariant}
U(R,a)\eta_\mu U(R,a)^{-1}=R_{\mu\nu} \eta^\nu + a_\mu I,\\
U(R,a)\cK_{\mu\nu} U(R,a)^{-1}=R_{\mu\rho} R^T_{\nu\sigma} \cK^{\rho \sigma},\nonumber
\end{gather}
as desired. We recall that the electric and magnetic parts
\begin{gather*}
\tilde{e}=\left(\tilde{A}_{01},\tilde{A}_{02},\tilde{A}_{03}\right)=2\left(f\circ H(\tilde{t}),0,0\right)\\
\tilde{m}=\left(\tilde{A}_{23},\tilde{A}_{13},\tilde{A}_{12}\right)=2\left(0,0,f\circ H(\tilde{t})\right)
\end{gather*}
respectively of the four tensor $\tilde{A}$ by definition transform under rotations as a pseudo-vector and a vector. As a consequence, we may write
\begin{equation}\label{electmagnet}
\cK_{0r} \left(\phi_1\otimes \phi_2\right)= R_{rl}\tilde{e}_l\phi_1\otimes \phi_2,\qquad \cK_{ij}\left(\phi_1\otimes \phi_2\right)= R_{ks}\tilde{m}_s \phi_1\otimes \phi_2.
\end{equation}
for $i,j,k,l,r,s=1,2,3$ and $i\neq j\neq k$.

We now show that inequalities \eqref{eq:general1},\eqref{eq:general2} are satisfied by the operators $(t=\eta_0,\eta_i)$, $i=1,2,3$. To do this, we will show that the left hand sides there are greater than a quantity which is rotation invariant and hence constant on fibers over $SO(3)$. Then, we will use \ref{lem:frie} of Appendix 1.\\
Consider the Hilbert space $L^2(\bR^5, d\xi)\otimes L^2(SO(3),d\mu(R))$. It is naturally unitarily equivalent to the space $L^2(SO(3),L^2(\bR^5, d\xi),d\mu(R))$ of square integrable functions on $SO(3)$ with values in $L^2(\bR^5, d\xi)$. For the sake of clarity we restrict to vector states defined by the dense (in the unit ball of $L^2(\bR^5, d\xi)\otimes L^2(SO(3),d\mu(R))$) set of vectors $\phi_1(\xi)\otimes \phi_2(R)$, with $\phi_1\in L^2(\bR^5, d\xi)$ and $\phi_2\in L^2(SO(3),d\mu(R))$ of unit norm. The general case would follow by abstract arguments \cite{sakai}. Evaluated on the operators $\eta_\mu$, they look like
\begin{gather*}
\omega_{\phi_1,\phi_2}(\eta_\mu)= \int \abs{\phi_2(R)}^2  \phi_1(\xi)^*(R_{\mu\nu}\tilde{\eta}_\nu \phi_1)(\xi) d\xi d\mu(R)=\\
=\int \tilde{\omega}_{\phi_1}(R_{\mu\nu}\tilde{\eta}_\nu)\abs{\phi_2(R)}^2 d\mu(R),
\end{gather*}
where $\tilde{\omega}_{\phi_1}$ indicates the vector state on operators on $L^2(\bR^5, d\xi)$ defined by $\phi_1$. Concerning variances, a simple adaptation of an argument in \cite{dop} gives:
\begin{gather}
\Delta_{\omega_{\phi_1,\phi_2}}\!(\eta_\mu)\geq \int \Delta_{\tilde{\omega}_{\phi_1}}\!(R_{\mu\nu}\tilde{\eta}_\nu)\abs{\phi_2(R)}^2 d\mu(R),\nonumber\\
\Delta_{\omega_{\phi_1,\phi_2}}\!(\eta_\mu)\Delta_{\omega_{\phi_1,\phi_2}}\!(\eta_\nu)\geq \int \Delta_{\tilde{\omega}_{\phi_1}}\!(R_{\mu\rho}\tilde{\eta}_\rho)\Delta_{\tilde{\omega}_{\phi_1}}\!(R_{\nu\sigma}\tilde{\eta}_\sigma)\abs{\phi_2(R)}^2 d\mu(R)\label{variance}.
\end{gather}
Concerning \eqref{eq:general2}, write
\begin{gather*}
\int \sum_{1\leq j<k\leq 3}\Delta_{\tilde{\omega}_{\phi_1}}\!(R_{jl}\tilde{\eta}_m)\Delta_{\tilde{\omega}_{\phi_1}}\!(R_{km}\tilde{\eta}_m)\abs{\phi_2(R)}^2 d\mu(R)\geq\\
\int \sqrt{\sum_{1\leq j<k\leq 3}\Delta_{\tilde{\omega}_{\phi_1}}\!(R_{jl}\tilde{\eta}_m)^2 \Delta_{\tilde{\omega}_{\phi_1}}\!(R_{km}\tilde{\eta}_m)^2}\abs{\phi_2(R)}^2 d\mu(R)\geq\\
\geq\frac{1}{2}\int \sqrt{\sum_{1\leq k\leq 3} \abs{\tilde{\omega}_{\phi_1}\!(R_{kl}\tilde{m}_k)}^2}\abs{\phi_2(R)}^2 d\mu(R)=\\
=\frac{1}{2}\sqrt{\sum_{1\leq k\leq 3} \abs{\tilde{\omega}_{\phi_1}\!(\tilde{m}_k)}^2}\int\abs{\phi_2(R)}^2 d\mu(R) =\\ =\lambda_P^2\omega_{\phi_1,\phi_2}\!\left(f\circ H(\tilde{t})\right) = \lambda_P^2\tilde{\omega}_{\phi_1,\phi_2}\!\left(f\circ H(t)\right),
\end{gather*}
where the first equality follows from rotation invariance of $\sum_k \abs{\tilde{\omega}_{\phi_1}\!(\tilde{m}_k)}^2$, the second one from \eqref{eq:bang} and we used \eqref{electmagnet}. Furthermore, the last equality comes from $\abs{\tilde{\omega}_{\phi_1}\!(\tilde{m}_3)}=\lambda_P^2\tilde{\omega}_{\phi_1}\!(f\circ H(\tilde{t}))$ and $\tilde{m}_1=\tilde{m}_2=0$. The proof of \eqref{eq:general1} involves the operators $\tilde e$ and goes along the same lines. 
\end{proof}

\section{Conclusions}
In this paper we wrote down physically motivated STUR to be satisfied in a quantum 
flat Friedmann spacetime and provided a concrete realisation of it writing down operators satisfying an appropriate weaker version of the STUR themselves. We made use of two key ingredients: an appropriate generalisation of Penrose's isoperimetric inequality, and a coordinate system made out of proper lengths measured along a suitable tetrad axis. Several consequences were deduced from our STUR and in particular we have shown that the presence of a particle horizon should naturally lead to the existence of a maximal value for the Hubble rate (or equivalently for the matter density), thus providing an indication that quantum effects may rule out a pointlike big bang singularity. Another interesting feature of our model is that it indicates that for expanding spacetimes the Planck length $\lambda_P$ should not be considered a fundamental length: an effective one appears depending on the Hubble rate $H$, \textit{i.e.} on the cosmological era. Finally, we costructed a covariant concrete realisation of the corresponding quantum Friedmann spacetime in terms of operators on some Hilbert space. 

A future line of research is the construction of quantum spacetimes also for open hyperbolic and closed Friedmann universes. In particular, in the case of closed universes these quantum models will allow us to explore the role of non-commutative effects near the big crunch
\cite{L3}.

\section*{Appendix 1}

In this appendix we prove some existence results for noncommutative coordinates suitable for the different models (\textit{i.e.} classical solutions of the Einstein equations) considered so far and on which the conclusions of Section \ref{sec:quantum} are based. For the sake of clarity, we first present concrete realisations in terms of operators acting on some Hilbert space of the basic commutation relations
$$
[\tilde{t},\tilde{\eta}_j ]=2if_0\left(H(\tilde{t})\right),
$$
without bothering about unitary actions of the isometry groups of the corresponding spacetimes. These will come later by using a ``covariantisation'' trick \cite{Pia2}. Since the time variable will not be touched upon by this procedure, to emphasise that the whole construction only depends on $H$, we can and will discuss the implementation of time diffeomorphism covariance at this preliminary stage.

To begin with, we recall some basic facts and terminology concerning (densely defined unbounded) linear operators on Hilbert spaces (see \cite{reed}). An operator $T$ on the Hilbert space $\cH$ with dense domain $\cD(T)\subset \cH$ is closed if its graph $\{ (x,Tx) \colon x\in\cD(T)\}$ is closed as a subspace of $\cH\times\cH$. An operator $T$ is called closable if the closure of its graph is the graph of an operator, usually indicated by $\overline T$. Given a closable operator $T$, we say that a subspace $\cC\subset\cD(T)$ is a core for $T$ if $\overline{T\restriction_{\cC}}=T$. We define the adjoint $T^*$ of a densely defined $T$ on $\cH$ with scalar product $(\cdot,\cdot)$ by setting $\phi\in\cD(T^*)$ if $\cD(T)\ni\psi\to (\phi,T\psi)$ can be extended to a bounded linear functional on the whole $\cH$. In this case there is a unique $\chi\in\cH$ such that $(\chi, \psi)=(\phi,T\psi)$ and we put $\chi=T^*\phi$. An operator $S$ is symmetric if $S\subset S^*$, meaning that $\cD(T)\subset\cD(T^*)$ \textit{and} $T^*\!\restriction_{\cD(T)}=T$. A symmetric operator is always closable and $S\subset\overline{S}\subset S^*$. A closed symmetric operator $T$ is selfadjoint if $T=T^*$ (as is the case of inclusion, equality here includes domains). The von Neumann's basic criterion states that a closed symmetric operator $S$ is selfadjoint if and only if Ker$(S^*\pm iI)=\emptyset$.
If we indicate by $\Delta_\pm$ the subspaces of solutions of the equations $(T^*+ iI)\phi=0$ and $(T^*- iI)\phi=0$ respectively, the dimensions dim$\Delta_\pm$ are called defect indices and selfadjoint extensions of $T$ exist if and only if they are equal.  If we indicate by $\Delta_\pm$ the subspaces of solutions of the equations $(T^*+ iI)\phi=0$ and $(T^*- iI)\phi=0$ respectively, the dimensions dim$(\Delta_\pm)$ are called defect indices and selfadjoint extensions $T^V$ of $T$ exist if and only if they are equal. The extensions $T^V$ satisfy $T\subset T^V\subset T^*$, are in one to one correspondence with the partial isometries $V\colon \Delta_+\to\Delta_-$ with initial space $I(V)$, have domains $\cD(T^V)=\{\phi+\phi_+ +U\phi_+\;|\;\phi\in\cD(T), \phi_+\in \Delta_+ \}$ and take the form
$$
T^V(\phi+\phi_+ +V\phi_+)=T\phi +i\phi_+ -iV\phi_+, \qquad \phi\in\cD(T), \phi_+\in \Delta_+.
$$
When dealing with commutators $[S,T]$ of unbounded operators $S,T$ on some $\cH$, we will indicate by $\cD_{\textrm{comm}}$ the set of all $\phi\in\cH$ such that both $ST\phi$ and $TS\phi$ are defined and in $\cH$.

A prominent example of unbounded operators are (partial) differential operators acting on $L^(\bR^d)$. The standard approach  (see \cite{schmu}) is to consider a ($n$-th order) partial differential expressions $\cL = \sum_{|\alpha|\leq n} a_\alpha(x) D^{\alpha}$,
with $\alpha =(\alpha_1, \dots, \alpha_d)$, $|\alpha|=\sum_i^d \alpha_i$ and $\alpha_i$ positive nonzero integers, $a_\alpha(x)\in C^\infty(\bR^d)$ and $D^\alpha = i^{|\alpha|}\partial_{\alpha_1}\dots \partial_{\alpha_1}$. One can then also define the \textit{formal adjoint}
$\cL^+ = \sum_{|\alpha|\leq n} D^{\alpha} a_\alpha(x)$, and view both $\cL$ and $\cL^+$ as operators on $L^2(\bR^d)$ with domain $C^\infty_c(\bR^d)$. The closure of $\cL^+$ is called the \textit{minimal operator} $L_{\textrm{min}}$. The \textit{maximal operator} $L_{\textrm{max}}$ is defined as follows. Take $\cD(L_{\textrm{max}})$ as the set of $\phi\in L^2(\bR^d)$ such that for some $\psi \in L^2(\bR^d)$ the equality $\psi=\cL\phi$ holds \textit{in the sense of distributions} and set $L_{\textrm{max}}\phi = g$. With these definitions, one easily proves that $L_{\textrm{min}}^*= L_{\textrm{max}}$. It follows that $\cL^*= L_{\textrm{max}}$ if $\cL=\cL^+$ as operators on $C^\infty_c(\bR^d)$, in which case $\cL$ is said to be formally selfadjoint.

\begin{Lemma}\label{lem:infla}
 Let $\tilde{\sH} =L^2(\bR^3, d\xi)$, $\gamma\colon \bR\to\emph{Im}(\gamma)\subset\bR$ be a diffeomorphism on its image $\emph{Im} (\gamma)=(\gamma(-\infty),\gamma(+\infty))$ and $H\colon \emph{Im}(\gamma)\to \bR^+$ be a smooth positive function  such that
\begin{equation}\label{condition}
\lim_{\tilde{t}\to\gamma(\pm\infty)}\int_{\gamma(0)}^{\tilde{t}}\frac{1}{f_0\circ H} =+\infty.
\end{equation}
Then the operators ($[\cdot,\cdot]_+$ indicates the anticommutator)
\begin{gather}
\tilde{t}= \tilde{\eta}_0 =2\gamma(\xi_1),\quad \tilde{\eta}_{1}= \frac{1}{2}\left[\frac{(f_0\circ H)(\gamma(\xi_1))}{\gamma'(\xi_1)},i\partial_{\xi_1}\right]_+\label{partialnoncomm}\\
\tilde{\eta}_{2}= 2\gamma(\xi_1) + i\partial_{\xi_2}, \qquad \tilde{\eta}_{3}= i\partial_{\xi_3}.\nonumber
\end{gather}
are essentially selfadjoint on the common invariant core $C_c^\infty(\bR^3)$ (the smooth functions with compact support) and satisfy the commutation relations
\begin{equation}\label{eq:bang}
[c\tilde{t},\tilde{\eta}_1 ]= 2if_0\circ H(\tilde{t}), \qquad \qquad [\tilde{\eta}_1, \tilde{\eta}_2 ]= 2if_0\circ H(\tilde{t}),
\end{equation}
there. All other commutators vanish. Moreover, $\emph{sp}(\tilde{t})=\emph{Im}(\gamma)$ and $\emph{sp}(\tilde{\eta}_k)=\bR$ for $k=1,2,3$. Finally, there is a strongly continuous unitary representation $\tilde{U}$ of $\bR^3$ such that $\tilde{U}(a)\tilde{\eta}_kU(a)^*=\tilde{\eta}_k +a_k$ and $\tilde{U}(a)\tilde{t}U(a)=\tilde{t}$, $\tilde{U}C_c^\infty(\bR^4)\subset C_c^\infty(\bR^4)$, $a=(a_1,a_2,a_3)$.
\end{Lemma}
\begin{proof}
To start with, recall that for operators $T_1, T_2$ on $\cH_1,\cH_2$ essentially selfadjoint on the domains $\cD_1,\cD_2$ the operators $T_1\otimes T_2$ and $T_1\otimes I+I\otimes T_2$ on $\cH_1\otimes\cH_2$ are essentially selfadjoint on $\cD_1\otimes\cD_2$ \cite{reed}. Thus, $\tilde{t},\tilde{\eta}_2,\tilde{\eta}_3$ are essentially selfadjoint on $C_c^\infty(\bR^3)\supset \otimes_{i=1}^3 C_c^\infty(\bR)$ just because $\xi_1,\xi_2,\xi_3$ are so on $C_c^\infty(\bR)$ as operators on $L^2(\bR)$. Concerning $\tilde{\eta}_1$, this is a differential operator on $C_c^\infty(\bR)\subset L^2(\bR)$ and a simple integration by parts shows that its adjoint $\tilde{\eta}_{1}^*$ satisfies $\tilde{\eta}_{1}^*\restriction_{C_c^\infty(\bR)}=\tilde{\eta}_{1}$. Thus, $\tilde{\eta}_1$ is formally selfadjoint and $\tilde{\eta}_{1}^*$ is densely defined with domain included in the space of continuous functions having locally integrable derivatives. Thus, the condition Ker$(\tilde{\eta}_{1}^*\pm iI)=\emptyset$ for essential selfadjointness reduces to the differential equation
\begin{gather}
ih(\xi_1)\phi'+\frac{i}{2}h'(\xi_1)\phi\pm i\phi=0,\qquad h(\xi_1)=(f_0\circ H)(\gamma(\xi_1))/\gamma'(\xi_1).\label{eq:h}
\end{gather}
The solutions $\phi_{\pm}(\xi_1)= C h^{-1/2}\exp{\left(\mp \int_0^{\xi_1}1/h\right)}$, with $C\in\bR$, are not in $L^2(\bR)$ whenever both
\begin{gather*}
\lim_{\xi_1\to\pm\infty} \int_0^{\xi_1}\frac{d\xi_1}{h(\xi_1)} = \lim_{\xi_1\to\pm\infty}\int_0^{\xi_1}\frac{\gamma'(\xi_1)d\xi_1}{(f_0\circ H)(\gamma(\xi_1))} =\\
=\lim_{\xi_1\pm\infty} \int_{\gamma(0)}^{\gamma(\xi_1)}\frac{d\tilde{t}}{(f_0\circ H)(\tilde{t})}=
\lim_{\tilde{t}\to\gamma(\pm\infty)} \int_{\gamma(0)}^{\tilde{t}}\frac{d\tilde{t}}{(f_0\circ H)(\tilde{t})},
\end{gather*}
diverge. The invariance of $C_c^\infty(\bR^3)$ under the action of our operators is obvious. The operator $(f_0\circ H)(\tilde{t})$ being bounded since the function $f_0\circ H$ is bounded, it follows that the commutation relations \eqref{eq:bang} hold there.\\
The fact that sp$(\tilde{t})=\textrm{Im}(\gamma)$ comes directly from functional calculus. To prove that sp$(\tilde{\eta}_k)=\bR$, $k=1,2,3$, we show that we have commuting unitaries such that $U C_c^\infty(\bR)\subset C_c^\infty(\bR)$ and $U(a_k)\tilde{\eta}_k U(a_k)^*=\tilde{\eta}_k+a_k$, $a_k\in\bR$. Clearly, the multiplication operators $\exp{i\left(a_1\int^{\xi_1} 1/h +a_2\xi_2 +a_3\xi_3\right)}$ do the job.
\end{proof}

From Lemma \ref{fundamental} we know that $f_0\to\sqrt[8]{12}$ as $H\to0$ and $f_0\sim(c/H\lambda_P)^{2/3}$ as $H\to +\infty$. Thus in reasonable physical situations we may safely assume that the limit for $\tilde{t}\to +\infty$ does satisfy condition \eqref{condition} but unfortunately this is not the case for $\tilde{t}\to \tilde{t}_0=\gamma(-\infty)$. As a matter of fact,
$$
\lim_{\tilde{t}\to\gamma(-\infty)}\int_{\gamma(0)}^{\tilde{t}} H(\tilde{t})^{2/3}
$$
does not diverge, even for power law cosmologies.
However, as we now show this difficulty may be circumvented by a simple trick.
\begin{Lemma}\label{lem:frie}
Let $\tilde{\sH} =L^2(\bR^5, d\xi)$, $\gamma\colon
\bR\to\emph{Im}(\gamma)\subset\bR$ be a diffeomorphism on its image
$\emph{Im}(\gamma)=(\gamma(-\infty),\gamma(+\infty))$ and $H\colon
\emph{Im}(\gamma)\to \bR^+$ be a smooth strictly positive function. Then the operators
\begin{gather}
\tilde{t}= \tilde{\eta_0}=2\gamma(\xi_1),\quad \tilde{\eta}_{2}= 2\gamma(\xi_1) + i\partial_{\xi_4}, \quad \tilde{\eta}_{3}= i\partial_{\xi_5}, \label{partialnoncomm2}
\end{gather}
(here $h$ is the same function as in \eqref{eq:h}) are essentially selfadjoint on $C_c^\infty(\bR^5)$. As defined on the same domain, the (closure of the) operator
\begin{equation}
\tilde{\eta}_{1}= ih(\xi_1)\partial_{\xi_1} +\frac{i}{2}h'(\xi_1)-ih(\xi_2)\partial_{\xi_2} -\frac{i}{2}h'(\xi_2) +i\partial_{\xi_3},
\end{equation}
is symmetric and admits a selfadjoint extension $\tilde{\eta}_{1}^V$. Together, they satisfy the commutation relations \eqref{eq:bang} on the corresponding domain $\cD_{\emph{comm}}$.  Moreover, $\emph{sp}(\tilde{t})=\emph{Im}(\gamma)$ and $\emph{sp}(\tilde{\eta}_k)=\bR$ for $k=1,2,3$. Finally, there is a strongly continuous unitary representation $U$ of $\bR^3$ such that $\tilde{U}(a)\tilde{\eta}_kU(a)^*=\tilde{\eta}_k +a_k$ and $\tilde{U}(a)\tilde{t}U(a)=\tilde{t}$, $\tilde{U}C_c^\infty(\bR^5)\subset C_c^\infty(\bR^5)$ and $\tilde{U}(a)\tilde{\eta}_k\tilde{U}(a)^*=\tilde{\eta}_k +a_k$, $a=(a_1,a_2,a_3)$.
\end{Lemma}

\begin{proof}
The very same arguments at the beginning of the proof of Lemma \ref{lem:infla} show that we can restrict our attention to the differential operator $T=ih(\xi_1)\partial_{\xi_1} +\frac{i}{2}h'(\xi_1)-ih(\xi_2)\partial_{\xi_2} -\frac{i}{2}h'(\xi_2)$ with domain $C_c^\infty(\bR^2)\subset L^2(\bR^2)$. Integration by parts shows it is formally selfadjoint. Indicating by $\cL_1$ the same differential expression than $T$ but with derivatives acting in the sense of distributions, by the above discussion we know that its adjoint $T^*=\cL_1$  on its domain $\cD(T^*)= \{\phi\in L^2(\bR^2) \;|\;\exists\psi\in L^2(\bR^2)\; \textrm{with}\; \psi=\cL_1\phi\}$. By construction the defect indexes are equal and the operator $T$ always admits sefadjoint extensions (trivial ones, if it is already essentially selfadjoint on $C_c^\infty(\bR^2)$). The von Neumann equations $(T^*\pm iI)\phi=0$ take the form of the partial differential equations
$$
\left( ih(\xi_1)\partial_{\xi_1}+\frac{i}{2}h^\prime(\xi_1) -ih(\xi_2)\partial_{\xi_2}-\frac{i}{2}h(\xi_2)h^\prime(\xi_2)\pm i\right)\phi_\pm=0.
$$
A straightforward calculation shows that for any sufficiently regular $u\colon \bR\to\bR$ the functions
$$
\phi_\pm(\xi_1,\xi_2)= \frac{Ce^{\mp \frac{1}{2}\left(\int_0^{\xi_1}1/h-\int_0^{\xi_2}1/h\right)}}{\sqrt{h(\xi_1)h(\xi_2)}} u\left(\int_0^{\xi_1}1/h +\int_0^{\xi_2}1/h\right),\quad C\in\bR,
$$
are solutions. This is enough to see that the defect indexes are both infinite. 

Consider then the ``flip'' unitary operator $V :L^2(\bR^2)\to L^2(\bR^2)$ given by $(V\phi)(\xi_1,\xi_2)= \phi(\xi_2,\xi_1)$ for $\phi\in L^2(\bR^2)$. By construction $VC_c^\infty(\bR^2)\subset C_c^\infty(\bR^2)$ and $V\tilde{\eta}_1V^*=-\tilde{\eta}_1$, so that $V\Delta_+=\Delta_-$ and we can consider the extension $T^V$ on the domain $\sD(T^V)=\{\psi\in L^2(\bR^2)\colon \psi = \phi+\phi_+ +V\phi_+, \; \phi\in C_c^\infty(\bR^2), \phi_+\in \Delta_+\}$ where, being a restriction of $T^*$, it acts as $\cL_1$. It follows that the operator $\tilde{\eta}_1^V=T^V+i\partial_{\xi_3}$ is essentially selfadjoint on the domain $\sD(T^V)\otimes C_c(\bR)$. Since $\tilde{\eta}_0$ is a multiplication operator by a smooth function, we can use the product rule for derivation of distributions to obtain (recall that $f_0\circ H$ is bounded smooth by assumption)
\begin{gather*}
\tilde{\eta}_1^V \tilde{\eta}_0\phi= \cL_1 2\gamma(\xi_1)\phi= ih(\xi_1)\partial_{\xi_1}(2\gamma(\xi_1)\phi)+ ih'(\xi_1)\gamma(\xi_1)\phi-\\
- ih(\xi_2)\partial_{\xi_2}(2\gamma(\xi_1)\phi)-ih'(\xi_2)\gamma(\xi_1)\phi +i\partial_{\xi_3}(2\gamma(\xi_1)\phi)=\\
=2i(f_0\circ H)(\gamma(\xi_1))\phi + 2\gamma(\xi_1)\cL_1\phi = 2i(f_0\circ H)(\tilde{t})\phi + \tilde{\eta}_0 \tilde{\eta}_1^U \phi,
\end{gather*}
for any $\phi\in \cD_{\textrm{comm}}\subset L^2(\bR^5)$. It follows immediately that the commutation relations \eqref{eq:bang} are satisfied for such $\phi$. The inclusion $C^\infty_c(\bR^5)\subset \cD_{\textrm{comm}}$ is obvious.\\
The proof that sp$(\tilde{t})=\textrm{Im}(\gamma)$ and sp$(\tilde{\eta}_k)=\bR$, $k=1,2,3$ goes along the same lines than in Lemma \ref{lem:infla}, but now we choose the multiplication operators $\tilde{U}(a)=\exp{i\left(a_1\xi_3+ a_2\xi_4 +a_3\xi_5 \right)}$.
\end{proof}
 
 \section*{Acknowledgements} 
We would like to thank Sergio Doplicher for his constant support and inspiration, Giorgio Immirzi for several conversations and suggestions. The patience and help of Gerardo Morsella were invaluable.

\end{document}